\documentclass[pdflatex, referee, icol]{sn-jnl} % the format I need

\usepackage{natbib}

\usepackage{graphicx}%
\usepackage{multirow}%
\usepackage{amsmath,amssymb,amsfonts}%
\usepackage{amsthm}%
\usepackage{mathrsfs}%
\usepackage[title]{appendix}%
\usepackage{xcolor}%
\usepackage{textcomp}%
\usepackage{manyfoot}%
\usepackage{booktabs}%
\usepackage{algorithm}%
\usepackage{algorithmicx}%
\usepackage{algpseudocode}%
\usepackage{listings}%
\theoremstyle{thmstyleone}%
\newtheorem{theorem}{Theorem}%  meant for continuous numbers
\theoremstyle{thmstyletwo}%

\theoremstyle{thmstylethree}%

\begin{document}

\title[Article Title]{A New Forward Discriminant Analysis Framework Based On Pillai's Trace and ULDA}


\author*[1]{\fnm{Siyu} \sur{Wang}}\email{swang739@wisc.edu} % (ORCID: \href{https://orcid.org/0000-0000-0000-0000}{0000-0000-0000-0000})

\author[1]{\fnm{Kehui} \sur{Yao}}\email{kyao24@wisc.edu}

\affil[1]{\orgdiv{Department of Statistics}, \orgname{University of Wisconsin-Madison}, \orgaddress{\street{1300 University Ave}, \city{Madison}, \postcode{53706}, \state{WI}, \country{USA}}}

%\affil[2]{\orgdiv{Department of Statistics}, \orgname{University of Wisconsin-Madison}, \orgaddress{\street{1300 University Ave}, \city{Madison}, \postcode{53706}, \state{WI}, \country{USA}}}

%\affil[2]{\orgdiv{Department}, \orgname{Organization}, \orgaddress{\street{Street}, \city{City}, \postcode{10587}, \state{State}, \country{Country}}}
%
%\affil[3]{\orgdiv{Department}, \orgname{Organization}, \orgaddress{\street{Street}, \city{City}, \postcode{610101}, \state{State}, \country{Country}}}


\abstract{Linear discriminant analysis (LDA), a traditional classification tool, suffers from limitations such as sensitivity to noise and computational challenges when dealing with non-invertible within-class scatter matrices. Traditional stepwise LDA frameworks, which iteratively select the most informative features, often exacerbate these issues by relying heavily on Wilks' $\Lambda$, potentially causing premature stopping of the selection process. This paper introduces a novel forward discriminant analysis framework that integrates Pillai's trace with Uncorrelated Linear Discriminant Analysis (ULDA) to address these challenges, and offers a unified and stand-alone classifier. Through simulations and real-world datasets, the new framework demonstrates effective control of Type I error rates and improved classification accuracy, particularly in cases involving perfect group separations. The results highlight the potential of this approach as a robust alternative to the traditional stepwise LDA framework.}




\maketitle

%------------------------------------------------
\section{Introduction}

LDA seeks to find linear combinations of features that can best separate groups by maximizing the ratio of between-group variance to within-group variance. However, LDA is sensitive to noise variables and prone to overfitting. To address these issues, stepwise LDA is introduced, which iteratively adds or removes variables based on predefined inclusion and exclusion criteria. Various versions of stepwise LDA have been developed, ranging from stand-alone programs like DIRCRIM \citep{mccabe1975computations} and ALLOC-1 \citep{hermans1976manual}, to options within statistical packages such as BMDP \citep{dixon1990bmdp}, SPSS$^{\text{\textregistered}}$ \citep{ibm2021ibm}, and SAS$^{\text{\textregistered}}$ \citep{sas2014sas}. While specific implementations may differ in variable selection criteria, most follow a common framework discussed in \cite{jennrich1977stepwise}. Nonetheless, the heavy reliance on Wilks' $\Lambda$ presents several challenges, some of which can be mitigated by substituting it with Pillai's trace.\\

Traditional LDA relies on the inverse of the within-class scatter matrix, leading to computational issues when the matrix is non-invertible. In contrast, ULDA \citep{ye2005characterization, ji2008generalized} uses a different loss function based on the trace to solve this problem. Since both ULDA and Pillai's trace use trace-based criteria, it is logical to integrate them to develop a more robust stepwise LDA framework.\\

This paper is organized as follows: Section \ref{sec:Lambda} discusses the limitations of Wilks' lambda in the traditional stepwise LDA framework and highlights the computational and statistical challenges that arise. Section \ref{sec:Algorithm} introduces our proposed forward selection framework based on Pillai's trace and ULDA, focusing on the algorithmic advancements and theoretical properties developed to resolve the challenges discussed. In Section \ref{sec:Empirical}, we present empirical analyses, including simulations and real data analyses, to demonstrate the effectiveness of the proposed method in controlling Type I error rates and improving classification accuracy when Wilks' $\Lambda$ fails. We conclude in Section \ref{sec:Conclusion}.

%------------------------------------------------
%------------------------------------------------
\section{Problems with Wilks' $\Lambda$}
\label{sec:Lambda}

First, we briefly introduce the most widely used stepwise LDA framework. Suppose we have a data matrix $\mathbf{X} \in \mathbb{R}^{N \times M}$ with $N$ observations and $M$ features. Our response $\mathbf{y} \in \mathbb{R}^{N}$ is a factor vector containing $J$ classes. Let $\mathbf{x}_{ji} \in \mathbb{R}^{M}$ represent the $i$th observation from class $j$, $\bar{\mathbf{x}}_j \in \mathbb{R}^{M}$ be the mean vector for class $j$ derived from its $n_j$ instances, and $\bar{\mathbf{x}} \in \mathbb{R}^{M}$ denote the overall mean vector across all samples. $\mathbf{H}_B \in \mathbb{R}^{J \times M}$ and $\mathbf{H}_W \in \mathbb{R}^{N \times M}$ are defined as:

\begin{align}
\mathbf{H}_B &= \left[\sqrt{n_1}\left(\bar{\mathbf{x}}_1-\bar{\mathbf{x}}\right), \sqrt{n_2}\left(\bar{\mathbf{x}}_2-\bar{\mathbf{x}}\right), \dots, \sqrt{n_J}\left(\bar{\mathbf{x}}_J-\bar{\mathbf{x}}\right)\right]^T, \nonumber \\
\mathbf{H}_W &= \left[\left(\mathbf{x}_{11} - \bar{\mathbf{x}}_1\right), \dots, \left(\mathbf{x}_{1n_1} - \bar{\mathbf{x}}_1\right), \left(\mathbf{x}_{21} - \bar{\mathbf{x}}_2\right), \dots, \left(\mathbf{x}_{2n_2} - \bar{\mathbf{x}}_1\right), \dots, \left(\mathbf{x}_{Jn_J} - \bar{\mathbf{x}}_J\right)\right]^T.
\label{eq:HbHw}
\end{align}

Then, the between-class scatter matrix $\mathbf{S}_B$, within-class scatter matrix $\mathbf{S}_W$, and total scatter matrix $\mathbf{S}_T$ can be defined as:

\begin{alignat}{2}
\mathbf{S}_B &= \sum_{j=1}^J n_j\left(\bar{\mathbf{x}}_j-\bar{\mathbf{x}}\right)\left(\bar{\mathbf{x}}_j-\bar{\mathbf{x}}\right)^{\prime} & &= \mathbf{H}_B^{T}\mathbf{H}_B \nonumber \\
\mathbf{S}_W &= \sum_{j=1}^J \sum_{i=1}^{n_j}\left(\mathbf{x}_{ji}-\bar{\mathbf{x}}_j\right)\left(\mathbf{x}_{ji}-\bar{\mathbf{x}}_j\right)^{\prime} & &= \mathbf{H}_W^{T}\mathbf{H}_W \nonumber \\
\mathbf{S}_T &= \sum_{j=1}^J \sum_{i=1}^{n_j} \left(\mathbf{x}_{ji}-\bar{\mathbf{x}}\right)\left(\mathbf{x}_{ji}-\bar{\mathbf{x}}\right)^{\prime} & &= \mathbf{S}_B + \mathbf{S}_W.
\end{alignat}

Let $\mathbf{S}_T(1,2,\dots,p)$ and $\mathbf{S}_W(1,2,\dots,p)$ be the total and within-class scatter matrix with $p$ variables $\{\mathbf{x}^{(1)}, \mathbf{x}^{(2)}, \dots, \mathbf{x}^{(p)} \}$ added. Then the Wilks' $\Lambda$ is defined as:

\begin{equation}
\Lambda(1,2,\dots,p) = \frac{\det(\mathbf{S}_W(1,2,\dots,p))}{\det(\mathbf{S}_T(1,2,\dots,p))}. \label{eq:Wilks}
\end{equation}

After adding $\mathbf{x}^{(p+1)}$, we use partial Wilks' $\Lambda$ to evaluate its marginal effect:

\begin{equation}
\Lambda(p+1)=\frac{\Lambda(1,2, \dots, p, p+1)}{\Lambda(1,2, \dots, p)}. \label{eq:PartialWilks}
\end{equation}

The null hypothesis $H_0$ states that the variables $\{\mathbf{x}^{(1)}, \mathbf{x}^{(2)}, \dots, \mathbf{x}^{(p+1)} \}$ are from a multivariate normal distribution and are independent of the response $\mathbf{y}$. Unless otherwise specified, this $H_0$ will be assumed as the null hypothesis throughout the remainder of this paper. Under $H_0$, the partial $F$-statistic follows an $F$-distribution:

\begin{equation}
F=\frac{N-J-p}{J-1} \frac{1-\Lambda(p+1)}{\Lambda(p+1)} \sim F_{J-1, N-J-p}.
\end{equation}

In the $(p+1)$-th step, partial $F$-statistics are calculated for the remaining $M-p$ variables, and the variable with the largest $F$-statistic is selected. It will be included in the model if it meets specific inclusion criteria, such as $F \geq 4$, or if the corresponding $p$-value is below $\alpha$.\\

Following the addition of a variable, the deletion phase begins. With $p+1$ variables now in the model, $p+1$ new pairs of scatter matrices $(\mathbf{S}_{W_{i}}, \mathbf{S}_{T_{i}})$ are computed, each excluding one variable $\mathbf{x}^{(i)}$. The partial $F$-statistics are then calculated for each pair, and the variable associated with the smallest $F$-statistic is considered for removal if the exclusion criterion is satisfied (e.g., $F < 3.996$ in BMDP). This stepwise process continues until all variables have been added, or no further variables can be added or removed.\\

Next, we introduce three major drawbacks of using Wilks' lambda in the current stepwise LDA framework.

%------------------------------------------------
\subsection{Premature Stopping}
\label{subsec:LambdaStop}

When perfect linear dependency exists in the data matrix, we would expect $\frac{0}{0}$ on the right-hand side of equation \eqref{eq:Wilks}, causing errors in some stepwise LDA programs, such as \texttt{klaR::greedy.wilks} in R. Wilks' $\Lambda$ is not well-defined under perfect linear dependency, and to allow the stepwise selection to continue, a quick fix is to manually set it to 1, indicating no discrimination power.\\

We know from equation \eqref{eq:PartialWilks} that the partial $\Lambda$ is the ratio of two Wilks' $\Lambda$. Most programs will stop the stepwise LDA process when $\Lambda = 0$, as all subsequent partial $\Lambda$ calculations become $\frac{0}{0}$ and are thus ill-defined. However, stopping at $\Lambda = 0$ isn't always appropriate, as it indicates that one group of classes is perfectly separable from another, but it doesn't necessarily imply perfect separation of all classes in non-binary classifications, as shown in Figure~\ref{Fig:Wilks0}. After selecting $X_2$, Wilks' $\Lambda = 0$ since the within-class variance is zero on $X_2$, causing the stepwise selection to stop. It successfully separates class A from classes B and C but cannot distinguish class B from class C. Additionally, when multiple variables result in $\Lambda = 0$, only one is selected, leading to the potential waste of useful information contained in the remaining variables.

\begin{figure}[htbp]
  	\centering
	\includegraphics[width = 0.5\textwidth]{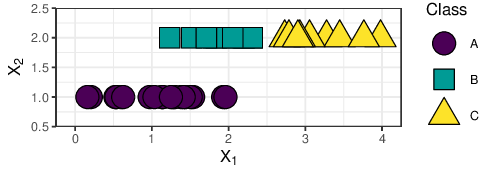}	
	\caption{A simulated pattern used in Section \ref{subsec:LambdaStop}. The stepwise LDA based on Wilks' $\Lambda$ selects $X_2$ and ignores $X_1$, leading to suboptimal performance}
	\label{Fig:Wilks0}
\end{figure}

%------------------------------------------------
\subsection{Partial $F$'s Distribution Under Stepwise Selection}
\label{subsec:LambdaPartial}

Here, we use an example to demonstrate that the distribution of the partial $F$-statistic does not follow an $F$-distribution under the stepwise selection framework, casting doubt on the validity of the associated hypothesis testing. Under $H_0$, the partial $F$-statistic follows an $F$-distribution \citep{rencher2002methods}. However, the original proof assumes that variables are ordered randomly, rather than selected through a stepwise process. Intuitively, if we maximize the $F$-statistic at each step, the result will be biased, as noted in \cite{rencher1980bias}.\\

The simulation setup is as follows: in each round, we simulate $N = 150$ observations from $J = 3$ classes, with each class having the same sampling probability of $1/3$. We simulate $X_1$ from a standard normal distribution in the one-variable scenario and simulate $X_1$ and $X_2$ from independent standard normal distributions in the two-variable scenario. We simulate 10,000 rounds and record the partial $F$-statistic from each round at each step. We then compare the simulated $F$-statistic with the theoretical distribution, as summarized in Figure~\ref{Fig:partialLambda}.\\

 The upper plot corresponds to the one-variable scenario. Since we have only one variable, $X_1$, no selection occurs, and the simulated distribution matches the theoretical distribution closely. The lower plot shows the two-variable scenario. Here, the stepwise selection first chooses $X^{(1)}$, which has a larger partial $F$-statistic (and a smaller partial $\Lambda$-statistic) compared to the other variable, resulting in an upward bias in the first partial $F$-statistic. The second partial $\Lambda$ is the ratio of the overall Wilks' $\Lambda$ (with two variables) to the first (partial) Wilks' $\Lambda$, so its partial $\Lambda$-statistic is biased upwards and its partial $F$-statistic is biased downwards. Note that the theoretical distributions for the first and second partial $F$ are different ($F_{2,147}$ and $F_{2,146}$), but the difference is negligible, so we assign the same color to both distributions in the plot.

\begin{figure*}[htbp]
  	\centering
	\includegraphics{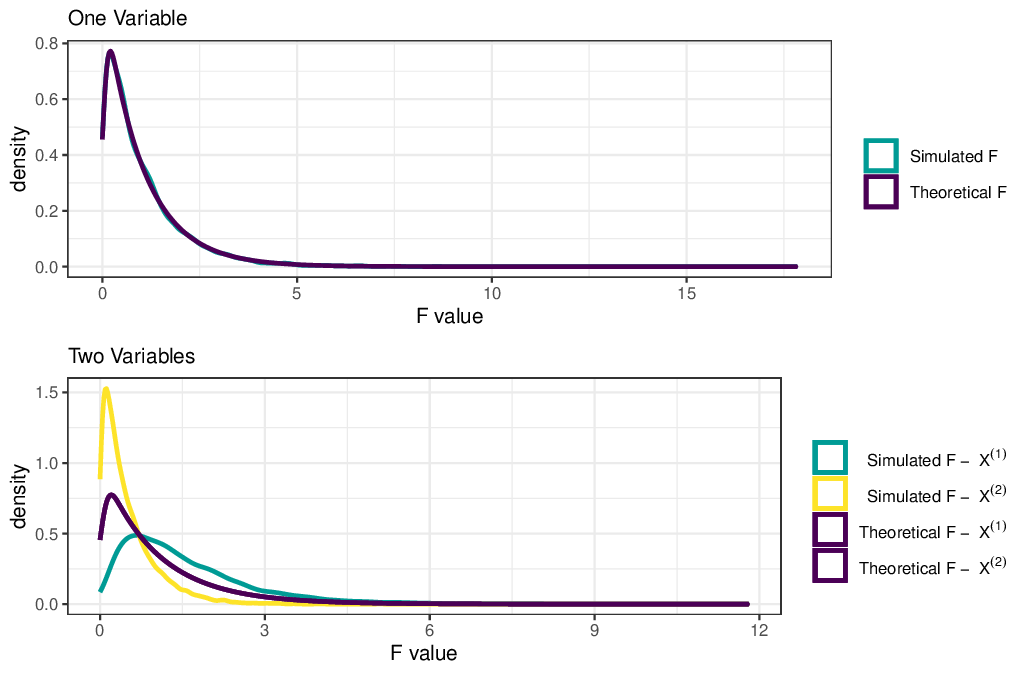}	
	\caption{The partial $F$-statistic does not follow an $F$-distribution under the stepwise selection framework (Section \ref{subsec:LambdaPartial}). When there are two variables (lower plot), the partial $F$ for the first variable $X^{(1)}$ is biased upwards, while the partial $F$ for the second variable $X^{(2)}$ is biased downwards}
	\label{Fig:partialLambda}
\end{figure*}

%------------------------------------------------
\subsection{Inflated Type I Error Rate}
\label{subsec:LambdaTypeI}

In most programs, a fixed threshold of $4$ is applied to the partial $F$-statistic. Another possible criterion is comparing the $p$-value of the partial $F$-statistic with the predefined $\alpha$. Here, we simulate two scenarios to demonstrate that the type I error is inflated since both methods fail to account for the number of variables screened. For simplicity, forward selection is used instead of stepwise selection throughout this paper unless otherwise stated.\\

We use the iris dataset for our first simulation. It contains $N = 150$ flowers from $J = 3$ species (50 setosa, 50 versicolor, and 50 virginica), along with four features that characterize the flowers. We then add $M$ mutually independent standard normal noise variables to it. Stepwise selections with both types of thresholds are performed, and we conclude that a type I error is made if at least one noise variable is selected. We let $M = 1, 2, 4, 8, 16, 32, 64, 128$, and for each $M$ we repeat the simulation 2,000 times to obtain a confidence band. $\alpha$ is set to $0.05$ throughout this paper unless otherwise specified.\\

In the second scenario, we simulate the null case where no variables are informative. We reuse the setup from the first scenario but remove all four original features from the iris dataset, leaving all remaining features independent of the species. The results from both scenarios are summarized in Figure~\ref{Fig:typeIwilks}. Both methods fail to control the type I error in either scenario. Due to the issue with multiple testing, the type I error rate increases as the number of noise variables grows. The fixed threshold of $4$ performs slightly better than using the $p$-value, partly because a threshold of $4$ corresponds to a $p$-value of approximately $0.02$ in this setting, based on $F_{2,147}(4) \approx F_{2,143}(4) \approx 0.02$.

\begin{figure*}[htbp]
  	\centering
	\includegraphics{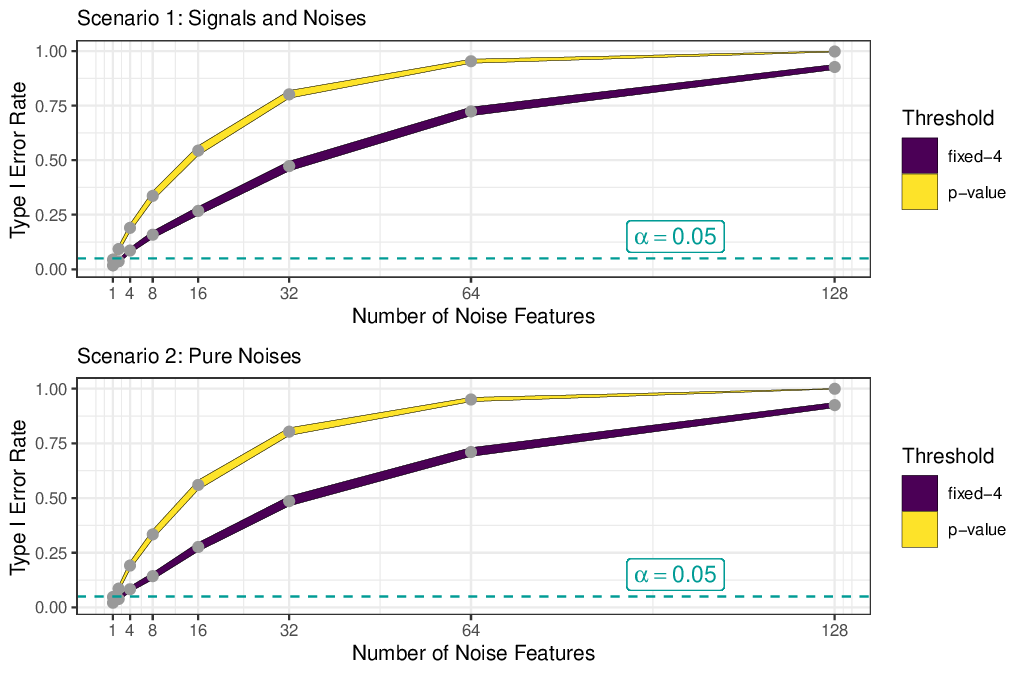}	
	\caption{The type I error rate of the forward stepwise selection is positively correlated with the number of noise features (Section \ref{subsec:LambdaTypeI}). Both types of thresholds fail to control it at the specified $\alpha$ level. The ribbon width represents the 95\% confidence interval}
	\label{Fig:typeIwilks}
\end{figure*}

%------------------------------------------------
\section{The Proposed Algorithm}
\label{sec:Algorithm}

In this section, we first introduce ULDA and compare it to classical LDA. Next, we present our enhancements to ULDA and make it a stand-alone classifier. We then introduce the forward ULDA framework and derive the distributions of the test statistics used. Finally, we demonstrate that the type I error rate is well controlled within the new framework and summarize the algorithm.

%------------------------------------------------
\subsection{ULDA: Overview and Enhancements}
\label{subsec:ULDA}

Uncorrelated LDA (ULDA) is an extension of LDA that addresses scenarios where the within-class scatter matrix, $\mathbf{S}_W$, is not invertible. Fisher's criterion aims to find transformation vectors $\mathbf{w} \in \mathbb{R}^{M}$ that maximizes the ratio:

\begin{equation}
\arg \max _{\mathbf{w}} \frac{\mathbf{w}^T \mathbf{S}_B \mathbf{w}}{\mathbf{w}^T \mathbf{S}_W \mathbf{w}}
\label{eq:fisherRatio}
\end{equation}

The optimal $\mathbf{W}$ is derived by solving an eigenvalue decomposition on $\mathbf{S}_W^{-1}\mathbf{S}_B$. The resulting eigen vectors $\mathbf{W} = [\mathbf{w}_1, \mathbf{w}_2, \dots]$ projects the original data $\mathbf{X}$ into orthogonal linear discriminant scores $\mathbf{X}\mathbf{w}_i$, ranked in descending order of their signal-to-noise ratios (eigenvalues). However, challenges arise when $\mathbf{S}_W$ is not invertible, such as when there are more variables than observations or when variables are linearly dependent. On the other hand, ULDA uses a different criterion:
\begin{equation}
		\mathbf{W}=\arg \max _{\mathbf{W}^T \mathbf{S}_T \mathbf{W}=I} \operatorname{trace}\left(\left(\mathbf{W}^T \mathbf{S}_T \mathbf{W}\right)^{+}\left(\mathbf{W}^T \mathbf{S}_B \mathbf{W}\right)\right),
		\label{eq:ULDA}
\end{equation}

where $\mathbf{A}^{+}$ denotes the Moore-Penrose inverse of matrix $\mathbf{A}$. This ensures that the ULDA solution always exists, and \cite{ye2004feature} shows that ULDA is equivalent to classical LDA when $\mathbf{S}_T$ is nonsingular. Several insights can be drawn from equation \eqref{eq:ULDA}:

\begin{itemize}
	\item $\mathbf{W}^T \mathbf{S}_T \mathbf{W} = I$, meaning we deliberately discard the null space of $\mathbf{S}_T$. This approach is reasonable because the null space of $\mathbf{S}_T$ is the intersection of the null spaces of $\mathbf{S}_W$ and $\mathbf{S}_B$, and equation \eqref{eq:fisherRatio} is not affected by vectors from the null space of $\mathbf{S}_B$.
	\item If we change the constraint from $\mathbf{W}^T \mathbf{S}_T \mathbf{W}=I$ to $\mathbf{W}^T\mathbf{W}=I$, it becomes Orthogonal LDA (OLDA) \citep{ye2005characterization}. However, the constraint $\mathbf{W}^T \mathbf{S}_T \mathbf{W}=I$ is more beneficial when constructing an LDA classifier. To solve equation \eqref{eq:ULDA}, \cite{howland2003structure} presents an algorithm based on Generalized Singular Value Decomposition (GSVD), where $\mathbf{S}_B$, $\mathbf{S}_W$, and $\mathbf{S}_T$ are diagonalized simultaneously. Suppose the rank of $\mathbf{S}_T$ is $M$, then 

\begin{align}
\mathbf{W}^T \mathbf{S}_B \mathbf{W} &= \mathrm{diag}(\alpha_1^2, \alpha_2^2, \dots, \alpha_M^2) \nonumber \\
\mathbf{W}^T \mathbf{S}_W \mathbf{W} &= \mathrm{diag}(\beta_1^2, \beta_2^2, \dots, \beta_M^2) \nonumber \\
\mathbf{W}^T \mathbf{S}_T \mathbf{W} &= \mathrm{diag}(\alpha_1^2 + \beta_1^2, \alpha_2^2 + \beta_2^2, \dots, \alpha_M^2 + \beta_M^2) \nonumber \\
&= \mathrm{diag}(1, 1, \dots, 1)
\end{align}
	
	Since the likelihood-based LDA classifier depends on the inverse of $\mathbf{S}_W$, computational resources are saved if it has already been diagonalized.
	\item Some commonly used test statistics from MANOVA are related to equation \eqref{eq:ULDA}. Pillai's trace is defined as $V = \operatorname{trace}\left(\mathbf{S}_T^{-1}\mathbf{S}_B\right)$, which equals $\sum_{i=1}^M \alpha_i^2$ after transformation. In other words, ULDA can be viewed as maximizing the generalized Pillai's trace under certain constraints. On the other hand, Wilks' $\Lambda = \prod_{i=1}^M \beta_i^2$. We discussed the issue with $\Lambda = 0$ in Section \ref{subsec:LambdaStop}. $\Lambda = 0$ means that $\beta_i = 0$ and $\alpha_i = 1$ for some $i$. In the stepwise selection framework, since Wilks' $\Lambda$ is a product, once it becomes zero, it remains zero. On the other hand, Pillai's trace is a summation, and adding another 1 has no side effect. Pillai's trace is also superior to Wilks' $\Lambda$ in other aspects \citep{rencher2002methods}.
\end{itemize}

Next, we introduce our speed enhancement for the ULDA algorithm when $N > M$. \cite{ye2005characterization} presents a ULDA algorithm that diagonalizes $\mathbf{S}_T$ and $\mathbf{S}_B$ separately. Based on our experience, it is slower by a constant factor compared to the GSVD-based version \citep{howland2003structure}, which is described in Algorithm \ref{alg:LDA/GSVD} (rewritten to suit our needs). However, when the sample size $N$ is large, the SVD decomposition on $\mathbf{K} \in \mathbb{R}^{(J+N) \times M}$ in the line~\ref{alg:LDA/GSVD:line2} of Algorithm~\ref{alg:LDA/GSVD} creates a runtime bottleneck. This can be resolved by reducing the dimension of $\mathbf{K}$ before performing SVD (or complete orthogonal decomposition). Since $\mathbf{H}_W$ contributes most of the dimensionality, and the SVD depends on $\mathbf{H}_W^{T}\mathbf{H}_W$, one possible solution is to replace $\mathbf{H}_W$ with $\mathbf{G}_W \in \mathbb{R}^{M \times M}$, where $\mathbf{H}_W^{T}\mathbf{H}_W = \mathbf{S}_W = \mathbf{G}_W^{T}\mathbf{G}_W$. We suggest performing a reduced QR decomposition $\mathbf{H}_W = \mathbf{Q}_W\mathbf{R}_W$ and replacing $\mathbf{H}_W$ with $\mathbf{R}_W$, so that we have
\begin{equation}
\mathbf{H}_W^{T}\mathbf{H}_W = \mathbf{R}_W^T\mathbf{Q}_W^T\mathbf{Q}_W\mathbf{R}_W = \mathbf{R}_W^T\mathbf{R}_W.
\end{equation}

\cite{park2007fast} follows a similar approach, using a Cholesky decomposition $\mathbf{S}_W = \mathbf{C}_W^{T}\mathbf{C}_W$ and replacing $\mathbf{H}_W$ with $\mathbf{C}_W$. We now use a simulation to evaluate the performance of these two variants and the original ULDA.\\

\begin{algorithm}[htbp]
\caption{ULDA via GSVD \citep{howland2003structure}}
\label{alg:LDA/GSVD}
\begin{algorithmic}[1]
\Require Data matrix $\mathbf{X} \in \mathbb{R}^{N \times M}$ with $N$ observations and $M$ features. Class label $\mathbf{y} \in \mathbb{R}^{N}$ containing $J$ classes
\Ensure Transformation matrix $\mathbf{W} \in \mathbb{R}^{M \times t_2}$ that preserves the class structure

\State Compute $\mathbf{H}_B \in \mathbb{R}^{J \times M}$ and $\mathbf{H}_W \in \mathbb{R}^{N \times M}$ from $\mathbf{X}$ according to equations \eqref{eq:HbHw}.
\State Compute the complete orthogonal decomposition of $\mathbf{K}=\left(\mathbf{H}_B^T, \mathbf{H}_W^T\right)^T \in \mathbb{R}^{(J+N) \times M}$ \label{alg:LDA/GSVD:line2}:
\begin{equation}
\mathbf{P}^T \mathbf{K} \mathbf{Q} = \begin{pmatrix}
\mathbf{R} & \mathbf{0} \\
\mathbf{0} & \mathbf{0}
\end{pmatrix}.
\end{equation}
\State Let $t_1 = \operatorname{rank}(\mathbf{K})$, $t_2 = \operatorname{rank}(\mathbf{P}(1: J, 1: t_1))$. Compute $\mathbf{V}$ from the SVD of $\mathbf{P}(1: J, 1: t_1)$, which satisfies: $\mathbf{U}^T \mathbf{P}(1: J, 1: t_1) \mathbf{V} = \mathrm{diag}(\alpha_1, \alpha_2, \dots, \alpha_{t_2}, 0, \dots, 0)$.
\State Compute the first $t_2$ columns of $\mathbf{Q}\mathbf{R}^{-1} \mathbf{V}$ and assign them to $\mathbf{W}$.

\end{algorithmic}
\end{algorithm}

The simulation setup is as follows: in each round, we simulate $N = 10000$ observations from $J = 10$ classes, with each class having the same sampling probability of $1/10$. The features are $M$ mutually independent standard normal noise variables. We then use the three algorithms to calculate the transformation matrix $\mathbf{W}$. We let $M = 2, 4, 8, 16, 32, 64, 128, 256, 512, 1024$, and for each $M$ we repeat the simulation 30 times to obtain a confidence band. The results are summarized in Figure~\ref{Fig:QR}. Their differences in runtime become larger as the number of features increases. For a data matrix of dimension $10000 \times 1024$, ULDA with QR decomposition is 38.6\% faster than the original GSVD-based ULDA implementation (6.2 seconds vs. 10.1 seconds). Consequently, we added this additional QR decomposition step into the ULDA pipeline.\\

\begin{figure*}[htbp]
  	\centering
	\includegraphics{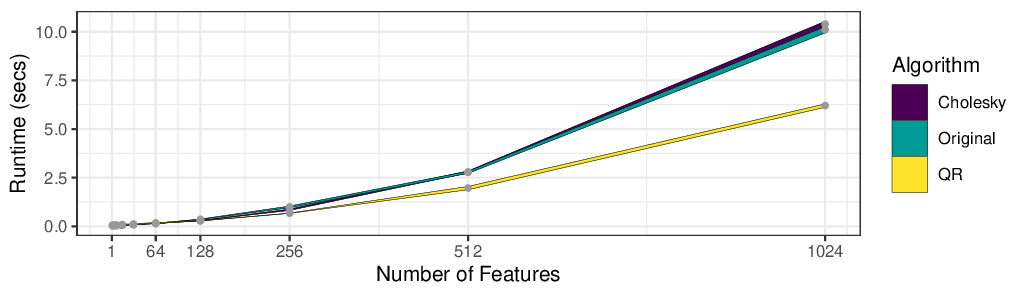}	
	\caption{The runtime comparisons of three ULDA implementations (Section \ref{subsec:ULDA}). The data consists of 10,000 observations from 10 classes. The ribbon width represents the 95\% confidence interval}
	\label{Fig:QR}
\end{figure*}

Another enhancement is our integration of the likelihood structure. The original ULDA method primarily serves as a dimension reduction tool and leaves the classification task to K-nearest-neighbors or other classifiers. However, LDA assumes Gaussian density for the features, and posterior probabilities naturally serve as a powerful classifier, as implemented in the R package \texttt{MASS} and the Python package \texttt{sklearn}. One difficulty lies in the invertibility of $\mathbf{S}_W$. Occasionally, we find a highly discriminative direction such that $\alpha^2 = 1$ (or Wilks' $\Lambda = 0$), meaning the ratio of the between-class variance to the total variance is 1, and one group of classes is perfectly separable from another. We address this by manually setting the within-class variance to $10^{-5}$, which has proven effective in our experience. To understand why it works, note that the discriminant function in LDA is:

\begin{equation}
\delta_j(\mathbf{x})=\mathbf{x}^T \boldsymbol{\Sigma}_W^{-1} \boldsymbol{\mu_j}-\frac{1}{2} \boldsymbol{\mu_j}^T \boldsymbol{\Sigma}_W^{-1} \boldsymbol{\mu_j}+\log \pi_j,
\end{equation}

where $\pi_j$ is the prior for class $j$, and we have $\widehat{\boldsymbol{\mu_j}} = \bar{\mathbf{x}}_j$ and $\widehat{\boldsymbol{\Sigma}_W} = \mathbf{S}_W / (N - J)$. $\mathbf{x}$ will be classified to class $k$ if $k = \arg\max_j \delta_j(\mathbf{x})$. When perfect separation occurs, we aim to capture this through the discriminant function. By setting the within-class variance to $10^{-5}$, the Mahalanobis norm associated with that direction is magnified $10^5$ times, which is large enough to dominate the effects from all other directions.\\

For missing values, we impute them with the median for numerical variables and assign a new level for categorical variables. Additionally, we include missing value indicators for numerical variables. For categorical variables, we use one-hot encoding and transform them into dummy variables. We also accommodate unequal misclassification costs. Let $\mathbf{C}$ represent the misclassification costs, where $C_{ij} = C(i \mid j)$ is the cost of classifying an observation into class $i$ given that it belongs to class $j$. Suppose the predicted posterior probability for an observation is $\widehat{\mathbf{p}} = (\hat{p}_1, \hat{p}_2, \dots, \hat{p}_j)$; then the cost of predicting it to class $i$ is

\begin{equation}
C_i = \sum_{j = 1}^J C(i \mid j) * \hat{p}_j,
\end{equation}

and $\mathbf{x}$ will be classified to class $k$ if $k = \arg\max_{i} C_i(\mathbf{x})$.

%------------------------------------------------
\subsection{Forward ULDA: Distribution and Threshold}
\label{subsec:FoLDA}

We now derive the distribution of the test statistics used in the forward ULDA framework. Without loss of generality, we assume $\mathbf{S}_T$ is invertible throughout this section, as redundant columns that cause $\mathbf{S}_T$ to be non-invertible have no discriminative power and can always be removed.

\begin{theorem}
\label{thm:PillaiMono}
Pillai's trace is non-decreasing when new variables are added to the model.
\end{theorem}

\begin{proof}
Suppose $\mathbf{X} \in \mathbb{R}^{N \times K}$ have been included in the model, and the new variable to be added is $\mathbf{z} \in \mathbb{R}^{N}$. The new between-class and total scatter matrices with $(K+1)$ variables can be written as:

\begin{equation}
\mathbf{S}_B =
\begin{pmatrix}
\mathbf{B}_x & \mathbf{b}_x \\
\mathbf{b}_x^{\prime} & b_z
\end{pmatrix}
\quad
\mathbf{S}_T =
\begin{pmatrix}
\mathbf{T}_x & \mathbf{t}_x \\
\mathbf{t}_x^{\prime} & t_z
\end{pmatrix},
\end{equation}

where $\mathbf{B}_x$ and $\mathbf{T}_x$ are the previous between-class and total scatter matrices for $\mathbf{X}$. If $K = 0$, the difference in Pillai's trace will be $b_z/t_z$. This value is non-negative since the between-class scatter matrix is positive semi-definite ($b_z \ge 0$) and the total scatter matrix is positive definite ($t_z > 0$). For $K \ge 1$, we aim to show

\begin{equation}
\operatorname{trace}\left(\mathbf{S}_T^{-1}\mathbf{S}_B\right) - \operatorname{trace}\left(\mathbf{T}_x^{-1}\mathbf{B}_x\right) \ge 0.
\end{equation}

Since $\mathbf{S}_T$ and $\mathbf{T}_x$ are invertible, we have

%\begin{equation}
%\mathbf{S}_T^{-1} =
%\begin{pmatrix}
%\mathbf{T}_x^{-1} + \mathbf{T}_x^{-1}\mathbf{t}_x\mathbf{t}_x^{\prime}\mathbf{T}_x^{-1}(t_z - \mathbf{t}_x^{\prime}\mathbf{T}_x^{-1}\mathbf{t}_x)^{-1} & -\mathbf{T}_x^{-1}\mathbf{t}_x(t_z - \mathbf{t}_x^{\prime}\mathbf{T}_x^{-1}\mathbf{t}_x)^{-1} \\
%-(t_z - \mathbf{t}_x^{\prime}\mathbf{T}_x^{-1}\mathbf{t}_x)^{-1}\mathbf{t}_x^{\prime}\mathbf{T}_x^{-1} & (t_z - \mathbf{t}_x^{\prime}\mathbf{T}_x^{-1}\mathbf{t}_x)^{-1}
%\end{pmatrix}.
%\label{eq:STinv}
%\end{equation}

\begin{equation}
\resizebox{\columnwidth}{!}{$
\mathbf{S}_T^{-1} =
\begin{pmatrix}
\mathbf{T}_x^{-1} + \mathbf{T}_x^{-1}\mathbf{t}_x\mathbf{t}_x^{\prime}\mathbf{T}_x^{-1}(t_z - \mathbf{t}_x^{\prime}\mathbf{T}_x^{-1}\mathbf{t}_x)^{-1} & -\mathbf{T}_x^{-1}\mathbf{t}_x(t_z - \mathbf{t}_x^{\prime}\mathbf{T}_x^{-1}\mathbf{t}_x)^{-1} \\
-(t_z - \mathbf{t}_x^{\prime}\mathbf{T}_x^{-1}\mathbf{t}_x)^{-1}\mathbf{t}_x^{\prime}\mathbf{T}_x^{-1} & (t_z - \mathbf{t}_x^{\prime}\mathbf{T}_x^{-1}\mathbf{t}_x)^{-1}
\end{pmatrix}.
$}
\label{eq:STinv}
\end{equation}

According to the block matrix multiplication and the properties of the trace, we have

\begin{align}
\operatorname{trace}\left(\mathbf{S}_T^{-1}\mathbf{S}_B\right) &= \operatorname{trace}\left(\mathbf{T}_x^{-1}\mathbf{B}_x\right) \nonumber \\
&\quad + (\mathbf{t}_x^{\prime}\mathbf{T}_x^{-1}, -1)\mathbf{S}_B(\mathbf{t}_x^{\prime}\mathbf{T}_x^{-1}, -1)^{\prime} \nonumber \\
&\quad \times (t_z - \mathbf{t}_x^{\prime}\mathbf{T}_x^{-1}\mathbf{t}_x)^{-1}.
\label{eq:traceDiff}
\end{align}

$(t_z - \mathbf{t}_x^{\prime}\mathbf{T}_x^{-1}\mathbf{t}_x)^{-1}$ is the Schur complement of the block $\mathbf{T}_x$ of the matrix $\mathbf{S}_T$. Since $\mathbf{T}_x$ and $\mathbf{S}_T$ are both positive definite, we have $(t_z - \mathbf{t}_x^{\prime}\mathbf{T}_x^{-1}\mathbf{t}_x)^{-1} > 0$. $(\mathbf{t}_x^{\prime}\mathbf{T}_x^{-1}, -1)\mathbf{S}_B(\mathbf{t}_x^{\prime}\mathbf{T}_x^{-1}, -1)^{\prime}$ is a quadratic form, and since the middle matrix is positive semi-definite, it is non-negative.
\end{proof}

At its core, ULDA seeks to maximize Pillai's trace $V = \operatorname{trace}\left(\mathbf{S}_T^{-1}\mathbf{S}_B\right)$. According to Theorem~\ref{thm:PillaiMono}, with each variable added, the current Pillai's trace increases (or remains the same). Let $V^{(M)}$ denote the Pillai's trace with all $M$ variables included. The goal of forward selection is to approximate $V^{(M)}$ using $V^{(K)}$, where $K \ll M$.\\

Suppose the variable set $\{\mathbf{x}_{1}, \mathbf{x}_{2}, \dots, \mathbf{x}_{K-1} \}$ has been selected after the first $(K-1)$ steps, and the Pillai's trace of that variable set is $V^{(K-1)}_{\max}$. Here, the subscript indicates that this Pillai's trace is not of $(K-1)$ randomly selected variables but is instead maximized at each step through forward selection. At step $K$, we calculate $V^{(K)}_{(1)}, V^{(K)}_{2}, \dots, V^{(K)}_{(M-K+1)}$, where $V^{(K)}_{(i)}$ denotes the Pillai's trace of the variable set $\{\mathbf{x}_{1}, \mathbf{x}_{2}, \dots, \mathbf{x}_{K-1}, \mathbf{x}_{(i)}\}$. Let $k = \arg\max_i V^{(K)}_{(i)}$. We then select $\mathbf{x}_{(k)}$ as the best candidate at step $K$, and $V^{(K)}_{\max} = V^{(K)}_{(k)}$. To establish an inclusion criterion, we must measure the marginal effect of the added variable $\mathbf{x}_{(k)}$, which corresponds to $V^{(K)}_{\max} - V^{(K-1)}_{\max}$.

\begin{theorem}
\label{thm:alpha}
At step $K$, let $t_K$ be the $(1 - \alpha)^{1/(M-K+1)}$ quantile of $\mathrm{B}\left(\frac{J-1}{2}, \frac{N-J}{2}\right)$. Then, $P(V^{(K)}_{\max} - V^{(K-1)}_{\max} \ge t_K) \le \alpha$ as $N \to \infty$ under $H_0$, where the newly added variable $\mathbf{z}$ is normally distributed and independent of both $\mathbf{X}$ and $\mathbf{y}$.
\end{theorem}

\begin{proof}
When $K = 1$, $V^{(1)}$ under $H_0$ follows a beta distribution $\mathrm{B}(\frac{J-1}{2}, \frac{N-J}{2})$ \citep{pillai1955some}. Since $V^{(1)}_{\max} - V^{(0)}_{\max} = V^{(1)}_{\max}$ is the maximum of $V^{(1)}_1, V^{(1)}_2, \dots, V^{(1)}_M$, which are $M$ i.i.d. random variables from the beta distribution, its CDF can be written as $I_{x}^{M}(\frac{J-1}{2}, \frac{N-J}{2})$ where $I_{x}(\frac{J-1}{2}, \frac{N-J}{2})$ is the CDF of $\mathrm{B}(\frac{J-1}{2}, \frac{N-J}{2})$. To control the type I error below $\alpha$, the threshold $t$ must satisfy $I_{t}^{M}(\frac{J-1}{2}, \frac{N-J}{2}) \le 1 - \alpha$, which is equivalent to $I_{t}(\frac{J-1}{2}, \frac{N-J}{2}) \le (1 - \alpha)^{1/M}$. Then $t$ is the $(1 - \alpha)^{1/M}$ quantile of $\mathrm{B}\left(\frac{J-1}{2}, \frac{N-J}{2}\right)$. \\

When $K > 1$, from equation \eqref{eq:traceDiff} we know that

\begin{align}
V^{(K)}_{(i)} - V^{(K-1)}_{\max} &= (\mathbf{t}_x^{\prime}\mathbf{T}_x^{-1}, -1)\mathbf{S}_B(\mathbf{t}_x^{\prime}\mathbf{T}_x^{-1}, -1)^{\prime} \nonumber \\
&\quad \times (t_z - \mathbf{t}_x^{\prime}\mathbf{T}_x^{-1}\mathbf{t}_x)^{-1}.
\label{eq:traceDiff2}
\end{align}

This equation still holds if we replace $\mathbf{S}_B$ and $\mathbf{S}_T$ with $\mathbf{S}_B/(N-J)$ and $\mathbf{S}_T/(N-J)$, which are the least squares estimators of the between-class and total covariance matrices. Since $\mathbf{X}$ and $\mathbf{z}$ are independent, their covariance $\mathbf{t}_x \to \mathbf{0}$ as $N \to \infty$. Note that $\mathbf{S}_B/(N-J)$ and $\mathbf{S}_T/(N-J)$ are finite as $N \to \infty$. Substituting $\mathbf{t}_x = \mathbf{0}$ into equation \eqref{eq:traceDiff2}, we get

\begin{align}
V^{(K)}_{(i)} - V^{(K-1)}_{\max} &= (\mathbf{0}, -1)\mathbf{S}_B(\mathbf{0}, -1)^{\prime}(t_z - 0)^{-1} \nonumber \\
&= b_z/t_z
\label{eq:traceDiff3}
\end{align}

$b_z/t_z$ is Pillai's trace for $\mathbf{z}$. Therefore, the distribution of $V^{(K)}_{(i)} - V^{(K-1)}_{\max}$ can be approximated by $V^{(1)}$, and the rest follows the scenario where $K = 1$.
\end{proof}

Based on our experience, this asymptotic approximation sometimes leads to a very conservative threshold, with the type I error falling well below the predefined $\alpha$. Therefore, we introduce an empirical approximation to mitigate this problem and achieve higher power. Suppose we have already added $K-1$ variables and the current Pillai's trace is $V^{(K-1)}_{\max}$. Since the maximum Pillai's trace for $J$ classes is $J-1$, the maximum Pillai's trace that can be added is bounded by $J - 1 - V^{(K-1)}_{\max}$, which can be viewed as the maximum Pillai's trace for a classification problem with $J - V^{(K-1)}_{\max}$ classes. Thus, at the $k$-th step, the threshold becomes the quantile from $\mathrm{B}\left(\frac{J^{\prime}-1}{2}, \frac{N-J^{\prime}}{2}\right)$ instead of $\mathrm{B}\left(\frac{J-1}{2}, \frac{N-J}{2}\right)$, where $J^{\prime} = J - V^{(K-1)}_{\max}$.

%------------------------------------------------
\subsection{Type I Error: Analysis and Control}
\label{subsec:typeI}

Here, we analyze the type I error under the forward ULDA framework and demonstrate that the family-wise error rate is controlled at the nominal level $\alpha$. Suppose we have $M$ variables in total, some of which are noise variables ($\mathbf{x} \in S_n$) and some are informative ($\mathbf{x} \in S_i$). At each step, there are three possible outcomes: a noise variable is selected, the selection stops, or an informative variable is selected. The entire process is illustrated in Figure~\ref{Fig:typeItree}.\\

\begin{figure}[htbp]
\centering
\includegraphics{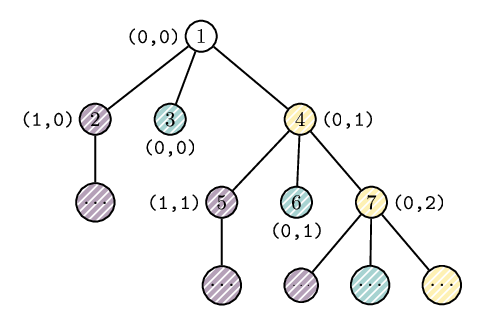}
\caption{Illustration of forward selection (Section \ref{subsec:typeI}). $(n_n, n_i)$ next to each node represents the number of noise and informative variables included so far. A yellow background indicates an informative variable being selected, a green background indicates the selection stops, and a purple background indicates that at least one type I error is made}
\label{Fig:typeItree}
\end{figure}

Suppose $\mathbf{x}$ is the variable with the largest Pillai's trace and is selected at the $K$-th step. Conditional on whether $\mathbf{x} \in S_n$ or $\mathbf{x} \in S_i$, there are four possible outcomes:

\begin{align}
	p_{K1} &= P(\mathbf{x} \text{ is added} \mid \mathbf{x} \in S_n) \nonumber \\
	p_{K2} &= P(\mathbf{x} \text{ is not added} \mid \mathbf{x} \in S_n) \nonumber \\
	p_{K3} &= P(\mathbf{x} \text{ is added} \mid \mathbf{x} \in S_i) \nonumber \\
	p_{K4} &= P(\mathbf{x} \text{ is not added} \mid \mathbf{x} \in S_i)
\end{align}

In situations with $p_{K2}$ and $p_{K4}$, $\mathbf{x}$ is not added, and forward selection stops. Therefore, no type I error is made or will be made. This corresponds to the green regions in Figure~\ref{Fig:typeItree}. For the situation with $p_{K3}$, since an informative variable is added, no type I error is made at the current step, corresponding to the yellow regions in Figure~\ref{Fig:typeItree}. The purple region in Figure~\ref{Fig:typeItree} reflects scenarios where a type I error is made, with $p_{K1}$ being the only situation that results in such an error. Theorem~\ref{thm:alpha} shows that under $H_0$, $p_{K1} \le \alpha$, meaning that at each step, the probability of branching into the purple region is controlled at $\alpha$. Now, we aim to show that, overall, the probability of ending up in any purple region is controlled at $\alpha$.\\

The probability of reaching node 2 is $p_{11} \le \alpha$. The probability of reaching node 5 is $p_{13} \times p_{21} \le p_{21} \le \alpha$. The reason we can use the product of $p_{13}$ and $p_{21}$ to calculate this probability is that under $H_0$, the variable selected in the first step is assumed to be independent of the variable selected in the second step. For nodes like node 2 and node 5, where the first noise variable is added in the current step, the probability of reaching them can be written as

\begin{equation}
p_{K1}\prod_{k = 1}^{K-1}p_{k3} \le p_{K1} \le \alpha.
\end{equation}

Meanwhile, the probability of reaching their child nodes is also controlled at $\alpha$, because reaching these nodes requires first reaching their parent node. All purple nodes fall into one of these two scenarios, so the family-wise type I error rate is controlled at $\alpha$. This means that if the forward ULDA selects $K$ variables $\{\mathbf{x}_{(1)}, \mathbf{x}_{(1)}, \dots, \mathbf{x}_{(K)}\}$, then the probability that at least one $\mathbf{x}_{(i)}$ is a noise variable is controlled at $\alpha$.\\

The forward selection framework is summarized in Algorithm \ref{alg:FoLDA}.

\begin{algorithm}[htbp]
\caption{Proposed forward selection based on Pillai's trace}
\label{alg:FoLDA}
\begin{algorithmic}[1]
\Require Data matrix $\mathbf{X} \in \mathbb{R}^{N \times M}$ with $N$ observations and $M$ features. Class label vector $\mathbf{y} \in \mathbb{R}^{N}$ containing $J$ classes
\Ensure A reduced data matrix $\mathbf{X}_{\text{new}} \in \mathbb{R}^{N \times K}$ where $K \le M$

\State index\_in $\Leftarrow \{\}$
\State index\_pool $\Leftarrow \{1,2,\ldots,M\}$
\State previous\_pillai $\Leftarrow 0$
\While{index\_pool is not empty}
    \State $l \Leftarrow$ length(index\_pool)
    \For{$i \Leftarrow 1$ to $l$}
        \State pillai\_saved[$i$] $\Leftarrow$ \texttt{get\_pillai}($\mathbf{X}$, $\mathbf{y}$, (index\_in, index\_pool[$i$])) \Comment{Calculate Pillai's trace for each feature}
    \EndFor
    \State i\_selected $\Leftarrow \arg\max_i$ pillai\_saved[$i$]
    \State $J^{\prime} = J - $ previous\_pillai
    \State threshold $\Leftarrow I_{(1-\alpha)^{1/l}}^{-1}(\frac{J^{\prime}-1}{2}, \frac{N-J^{\prime}}{2})$
    \If{pillai\_saved[i\_selected] - previous\_pillai $\le$ threshold}
        \State \textbf{break} \Comment{Stop if the improvement is below threshold}
    \Else
        \State index\_selected $\Leftarrow$ index\_pool[i\_selected]
        \State index\_in $\Leftarrow$ (index\_in, index\_selected)
        \State index\_pool $\Leftarrow$ index\_pool $\setminus$ \{index\_selected\}
        \State previous\_pillai $\Leftarrow$ pillai\_saved[i\_selected]
    \EndIf
\EndWhile
\If{length(index\_in) $= 0$}
    \State $\mathbf{X}_{\text{new}} \Leftarrow \mathbf{X}$ \Comment{No variable is significant, return all variables}
\Else
    \State $\mathbf{X}_{\text{new}} \Leftarrow \mathbf{X}$(:, index\_in)
\EndIf
\end{algorithmic}
\end{algorithm}

%------------------------------------------------
\section{Empirical Analysis}
\label{sec:Empirical}

In this section, we use simulation and real data to showcase the performance of three forward LDA variants:

\begin{enumerate}
	\item \texttt{Pillai}: the proposed variant using Pillai's trace (see Algorithm \ref{alg:FoLDA}). 
	\item \texttt{Wilks}: the original variant based on Wilks' $\Lambda$. The inclusion criterion is based on $p$-value, with a variable included if the $p$-value is below the predefined $\alpha$. 
	\item \texttt{Wilks-Bonferroni}: This variant applies an additional Bonferroni correction to the $p$-value compared to the second variant. If there are $(M-K+1)$ variables to choose from at the $K$-th step, the $p$-value is multiplied by $(M-K+1)$ to adjust for the multiple testing.
\end{enumerate}

Note that these variants are for selection alone. They help check the type I error and power (whether the desired variables are included). To further compare the testing accuracy, we apply ULDA (see Section \ref{subsec:ULDA}) to the selected variables.

%------------------------------------------------
\subsection{Type I Error Evaluation on Iris: Pure Noise and Mixed Cases}
\label{subsec:typeIemp}

We use the same simulation settings from Section \ref{subsec:LambdaTypeI} to compare the three forward LDA variants, and the results are summarized in Figure~\ref{Fig:typeIwilks3}. \texttt{Pillai} and \texttt{Wilks-Bonferroni} successfully control the type I error in both scenarios. In contrast, \texttt{Wilks} suffers from an inflated type I error rate due to multiple testing. These results validate Theorem~\ref{thm:alpha}, demonstrating that the type I error rate is well-controlled under $H_0$, where the noise variables are normally distributed and independent of both the informative variables and the response.
\begin{figure*}[htbp]
  	\centering
	\includegraphics{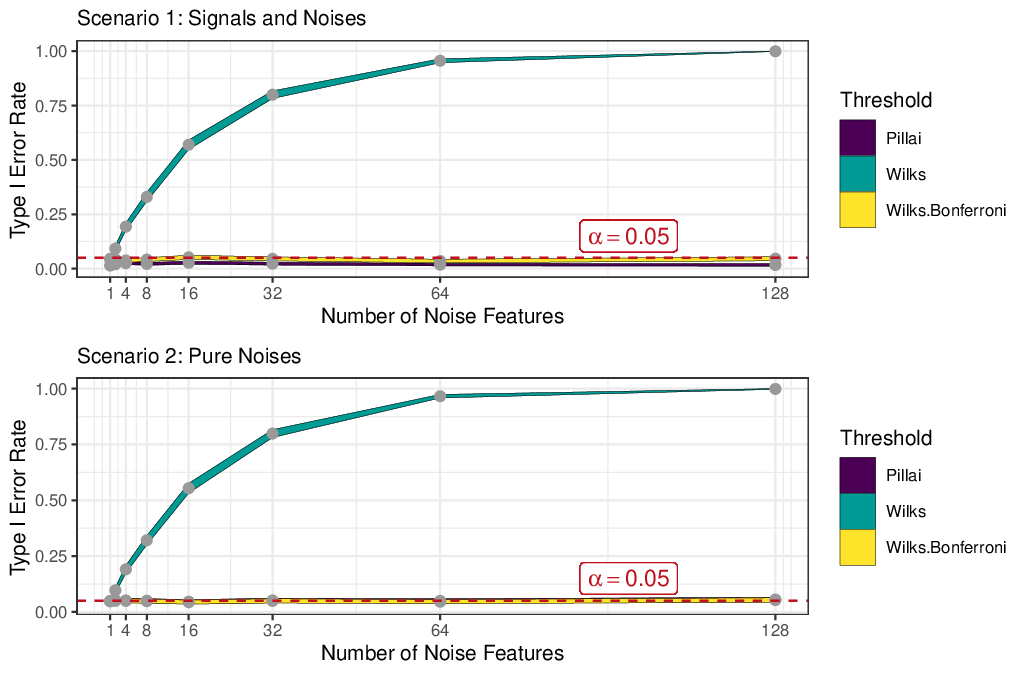}	
	\caption{Type I error rate from three forward LDA variants on the iris dataset (Section \ref{subsec:typeIemp}). The ribbon width represents the 95\% confidence interval. The ribbon for \texttt{Pillai} overlaps with the ribbon for \texttt{Wilks-Bonferroni} in the lower plot}
	\label{Fig:typeIwilks3}
\end{figure*}

%------------------------------------------------
\subsection{Handling $\Lambda=0$: Analysis on Simulated and Real Data}
\label{subsec:Wilks0emp}

In this section, we illustrate the primary advantage of our proposed method over the original framework: its ability to handle scenarios where Wilks' $\Lambda = 0$.\\

First, we use a simulated dataset, which contains 2,000 observations. The response variable is randomly selected from 10 classes, each with an equal probability of $1/10$. We then create a dummy matrix (one-hot encoding) of the response, resulting in 10 columns, each consisting of 1s or 0s. These 10 columns are used as our features. Ideally, these 10 features can perfectly predict the response, and a robust forward selection method should select them all. However, \texttt{Wilks} stops after selecting only one feature, $I_{\text{Class One}}$, the indicator of the first class. Here, $I_{\text{Class One}} = 1$ for observations from class one and $I_{\text{Class One}} = 0$ for the other classes. Therefore, the within-class variance is $0$, leading to Wilks' $\Lambda = 0$. Meanwhile, with \texttt{Pillai}, the feature $I_{\text{Class One}}$ contributes a value of 1 to the overall Pillai's trace, which does not trigger a stop. \texttt{Pillai} continues adding features, with each feature $I_{\text{Class i}}$ contributing a value of 1 to the overall Pillai's trace. It selects 9 features and then stops, as the maximum Pillai's trace of $J - 1 = 9$ is reached. With these 9 features, we have enough information to perfectly predict the response.\\

Next, we use real data from the National Highway Traffic Safety Administration (NHTSA) and its Vehicle Crash Test Database. This database contains data from various crash tests, including those conducted for research, the New Car Assessment Program (NCAP), and compliance purposes. The dataset captures a wide range of vehicle attributes, crash conditions, and safety outcomes. Our focus here is to predict the type of engine used in the tested vehicles. This dataset is challenging to analyze for several reasons:

\begin{itemize}
	\item Cyclic variables: Variables like impact angle are cyclic, where $359^\circ$ and $1^\circ$ should be very similar in real life but are very different in their numerical representation. We address this by transforming all angles to their $\cos$ and $\sin$ values.
	\item Missing values: For missing values, we impute them with the median for numerical variables and assign a new level for categorical variables. Additionally, we include missing value indicators for numerical variables.
	\item Multicollinearity: Variables BX1 to BX21 are measurements of different parts of the car, and some of them are highly correlated, which affects the performance of classical LDA, making forward ULDA more promising.
\end{itemize}

After data preprocessing, there are 3,273 crash tests and 173 (982 if all categorical variables are transformed into dummy variables) predictor variables. Our response variable, engine type, has 18 different types, with an imbalanced distribution: the most frequent class, 4CTF, has 1,250 occurrences, while the least frequent class, NAPP, has only one occurrence. We then apply both forward LDA variants, \texttt{Wilks} and \texttt{Pillai}, to this dataset.\\

\texttt{Wilks} stops after selecting variables $I_{\text{Model}=\text{RX}}$ and $I_{\text{Model}=\text{RX-8}}$. The reason is that only these two Mazda models (RX and RX-8) use the ROTR engine type. Therefore, the within-class variance becomes zero, allowing it to perfectly identify the ROTR engine. However, this premature stopping prevents it from identifying all the other engine types. The order of the variables selected also reveals its lack of statistical power. The first selected variable should be the most discriminative. However, the top two variables selected by \texttt{Wilks} can only correctly identify 3 instances out of 3,273, clearly not the most discriminative variables. This bias towards the perfect separation of certain classes, regardless of how few the instances are, dominates the forward selection process via smaller Wilks $\Lambda$, indicating that Wilks $\Lambda$ is not an ideal test statistic and fails to capture the most discriminative information across all classes.\\

In contrast, \texttt{Pillai} ends up selecting 175 variables. Using ULDA as the classifier and a 10-fold cross-validation, the prediction accuracies of \texttt{Wilks} and \texttt{Pillai} are 0.38 and 0.65, respectively. Results from simulation and real data demonstrate that \texttt{Pillai} outperforms \texttt{Wilks} when Wilks' $\Lambda = 0$.

%------------------------------------------------
\section{Conclusion}
\label{sec:Conclusion}

In this paper, we present a new forward discriminant analysis framework based on Pillai's trace, which demonstrates superiority in situations involving perfect group separations (Wilks' $\Lambda = 0$) compared to traditional methods. Our approach effectively controls the type I error rate and integrates seamlessly with ULDA, providing a unified classifier.\\

Despite these advancements, there are some limitations to consider. The criterion used in forward selection, whether based on Wilks' $\Lambda$ or Pillai's trace, serves only as a measure of goodness of fit within the MANOVA framework, as mentioned in \cite{rencher2002methods}. It does not directly reflect the model's training or testing accuracy, although these two typically align well. For those primarily concerned with accuracy, we recommend using the forward selection framework to rank variables, followed by methods such as cross-validation to choose the best subset of variables.\\

Additionally, our derivation of the distribution of the test statistic in Section \ref{subsec:FoLDA} is asymptotic and includes an empirical approximation. There is potential to develop a more precise distribution for finite sample size under alternative assumptions, which could be a direction for future research.\\

In this paper, we focus primarily on forward selection rather than stepwise selection. The main reason is the difficulty in theoretically justifying a valid exclusion criterion. Moreover, based on our experience, the addition of a variable deletion step typically results in minimal improvement in classification accuracy while significantly increasing runtime.\\

The related R package \texttt{folda} is available on CRAN.

\backmatter

%\bmhead{Supplementary information}
%
%If your article has accompanying supplementary file/s please state so here. 
%
%Authors reporting data from electrophoretic gels and blots should supply the full unprocessed scans for key as part of their Supplementary information. This may be requested by the editorial team/s if it is missing.
%
%Please refer to Journal-level guidance for any specific requirements.

%\bmhead{Acknowledgements}
%
%The author gratefully acknowledges Prof. Wei-Yin Loh from UW-Madison for his unwavering guidance throughout the author's PhD journey.

%Acknowledgements are not compulsory. Where included they should be brief. Grant or contribution numbers may be acknowledged.
%
%Please refer to Journal-level guidance for any specific requirements.

\section*{Declarations}

Conflict of interest: The authors declare no conflict of interest.

\bibliographystyle{plainnat}
\bibliography{sn-bibliography}% common bib file

\begin{thebibliography}{14}
\providecommand{\natexlab}[1]{#1}
\providecommand{\url}[1]{\texttt{#1}}
\expandafter\ifx\csname urlstyle\endcsname\relax
  \providecommand{\doi}[1]{doi: #1}\else
  \providecommand{\doi}{doi: \begingroup \urlstyle{rm}\Url}\fi

\bibitem[Dixon(1990)]{dixon1990bmdp}
Wilfrid~Joseph Dixon.
\newblock \emph{BMDP statistical software manual: to accompany the 1990
  software release}, volume~1.
\newblock University of California Press, 1990.

\bibitem[Hermans and Hobbema(1976)]{hermans1976manual}
J~Hermans and JDF Hobbema.
\newblock \emph{Manual for the ALLOC Discriminant Analysis Programs: A Package
  of FORTRAN Computer Programs}.
\newblock University of Leiden, Department of Medical Statistics, 1976.

\bibitem[Howland et~al.(2003)Howland, Jeon, and Park]{howland2003structure}
Peg Howland, Moongu Jeon, and Haesun Park.
\newblock Structure preserving dimension reduction for clustered text data
  based on the generalized singular value decomposition.
\newblock \emph{SIAM Journal on Matrix Analysis and Applications}, 25\penalty0
  (1):\penalty0 165--179, 2003.

\bibitem[{IBM Corp.}(2021)]{ibm2021ibm}
{IBM Corp.}
\newblock {IBM SPSS Statistics for Windows, Version 28.0}.
\newblock \emph{Armonk, NY: IBM Corp}, 2021.

\bibitem[Jennrich(1977)]{jennrich1977stepwise}
Robert~I Jennrich.
\newblock Stepwise discriminant analysis.
\newblock \emph{Statistical methods for digital conputers}, 76, 1977.

\bibitem[Ji and Ye(2008)]{ji2008generalized}
Shuiwang Ji and Jieping Ye.
\newblock Generalized linear discriminant analysis: a unified framework and
  efficient model selection.
\newblock \emph{IEEE Transactions on Neural Networks}, 19\penalty0
  (10):\penalty0 1768--1782, 2008.

\bibitem[McCabe(1975)]{mccabe1975computations}
George~P McCabe.
\newblock Computations for variable selection in discriminant analysis.
\newblock \emph{Technometrics}, 17\penalty0 (1):\penalty0 103--109, 1975.

\bibitem[Park et~al.(2007)Park, Drake, Lee, and Park]{park2007fast}
Haesun Park, Barry~L Drake, Sangmin Lee, and Cheong~Hee Park.
\newblock Fast linear discriminant analysis using qr decomposition and
  regularization.
\newblock Technical report, Georgia Institute of Technology, 2007.

\bibitem[Pillai(1955)]{pillai1955some}
KC~Sreedharan Pillai.
\newblock Some new test criteria in multivariate analysis.
\newblock \emph{The Annals of Mathematical Statistics}, pages 117--121, 1955.

\bibitem[Rencher and Christensen(2002)]{rencher2002methods}
Alvin~C Rencher and WF~Christensen.
\newblock Methods of multivariate analysis. a john wiley \& sons.
\newblock \emph{Inc. Publication}, 727:\penalty0 2218--0230, 2002.

\bibitem[Rencher and Larson(1980)]{rencher1980bias}
Alvin~C Rencher and Steven~F Larson.
\newblock {Bias in Wilks' $\Lambda$ in stepwise discriminant analysis}.
\newblock \emph{Technometrics}, 22\penalty0 (3):\penalty0 349--356, 1980.

\bibitem[{SAS Institute Inc.}(2014)]{sas2014sas}
{SAS Institute Inc.}
\newblock {SAS/ETS$^{\tiny{\textregistered}}$ 13.2 User’s Guide:
  High-Performance Procedures.}
\newblock \emph{Cary, NC: SAS Institute Inc.}, 2014.

\bibitem[Ye and Yu(2005)]{ye2005characterization}
Jieping Ye and Bin Yu.
\newblock Characterization of a family of algorithms for generalized
  discriminant analysis on undersampled problems.
\newblock \emph{Journal of Machine Learning Research}, 6\penalty0 (4), 2005.

\bibitem[Ye et~al.(2004)Ye, Janardan, Li, and Park]{ye2004feature}
Jieping Ye, Ravi Janardan, Qi~Li, and Haesun Park.
\newblock Feature extraction via generalized uncorrelated linear discriminant
  analysis.
\newblock In \emph{Proceedings of the twenty-first international conference on
  Machine learning}, page 113, 2004.

\end{thebibliography}

\end{document}